\iftrue
\documentclass[
10pt,
notitlepage,
superscriptaddress,
nofootinbib,
longbibliography,
floatfix,
onecolumn,
]{revtex4-1}
\else
\documentclass[10pt]{iopart}
\usepackage{iopams}
\expandafter\let\csname equation*\endcsname\relax
\expandafter\let\csname endequation*\endcsname\relax
\fi

\usepackage[utf8]{inputenc}
\usepackage[T1]{fontenc}
\usepackage[sc,osf]{mathpazo}
\usepackage{amsmath, amsthm, amssymb, amsfonts}
\usepackage{dsfont}
\usepackage{braket}
\usepackage{graphicx}
\usepackage{hyperref}

\hypersetup{breaklinks=true}
\hypersetup{citecolor=blue}
\hypersetup{colorlinks=true}
\hypersetup{urlcolor=blue}

\providecommand{\openone}{{\mathds 1}}
\newcommand{\reals}{{\mathds R}}
\newcommand{\soc}{{S}}
\newcommand{\M}{{\mathsf M}}
\newcommand{\CC}{{\mathcal C}}
\newcommand{\II}{{\mathcal I}}

\newcommand{\TT}{{\mathcal T}}
\newcommand{\IA}{{\mathcal X}}
\newcommand{\CA}{{\mathcal C}}
\newcommand{\OA}{{\mathcal A}}
\newcommand{\ST}{{\mathcal S}}

\newcommand{\ketbra}[1]{{\ket{#1}\!\bra{#1}}}
\newcommand{\KetBra}[2]{{\ket{#1}\!\bra{#2}}}
\newcommand{\norm}[1]{{\lvert{#1}\rvert}}
\newcommand{\abs}[1]{{\lvert{#1}\rvert}}

\newcommand{\vtr}[1]{{\boldsymbol #1}}

\newtheorem{lemma}{Lemma}

\DeclareMathOperator{\tr}{tr}

\begin{document}
\title{Memory cost of temporal correlations}

\author{Costantino Budroni}
\email{costantino.budroni@oeaw.ac.at}
\affiliation{Institute for Quantum Optics and Quantum Information (IQOQI), 
Austrian Academy of Sciences, Boltzmanngasse 3, 1090 Vienna, Austria}
\affiliation{Faculty of Physics, University of Vienna, Boltzmanngasse 5, 1090 Vienna, Austria}

\author{Gabriel Fagundes}
\email{gabrielf@fisica.ufmg.br}
\affiliation{Departamento de Física, Universidade Federal de Minas Gerais UFMG, 
P.O.~Box 702, 30123–970, Belo Horizonte, MG, Brazil}

\author{Matthias Kleinmann}
\email{matthias.kleinmann@uni-siegen.de}
\affiliation{Naturwissenschaftlich--Technische Fakultät, Universität Siegen, 
Walter-Flex-Straße 3, 57068 Siegen, Germany}

\begin{abstract}
A possible notion of nonclassicality for single systems can be defined on the 
basis of the notion of memory cost of classically simulating probabilities 
observed in a temporal sequence of measurements. We further explore this idea 
in a theory-independent framework, namely, from the perspective of general 
probability theories (GPTs), which includes classical and quantum theory as 
special examples. Under the assumption that each system has a finite memory 
capacity, identified with the maximal number of states perfectly 
distinguishable with a single measurement, we investigate what are the temporal 
correlations achievable with different theories, namely, classical, quantum, 
and GPTs beyond quantum mechanics. Already for the simplest nontrivial 
scenario, we derive inequalities able to distinguish temporal correlations 
where the underlying system is classical, quantum, or more general.
\end{abstract}

\maketitle

\section{Introduction}

Given a single quantum system, in what sense can we say that it has some 
nonclassical properties? The most celebrated phenomena where quantum systems 
depart from their classical counterpart involve notions such as entanglement 
\cite{ent_revH, ent_revG} and nonlocality \cite{BellOriginal, nonloc_rev}, 
which can be defined only in terms of multipartite systems. What if we are able 
to perform experiments only on a single, indivisible, system? Can we still say 
that the observed statistics has some ``nonclassical properties''? Some notion 
of nonclassicality have been proposed for single systems, such as contextuality 
\cite{Kochen1967} and nonmacrorealism \cite{LeggettPRL1985, EmaryRPP2014}. One 
may argue that such notions are limited to specific measurement procedures and 
hence are not fully satisfactory. Contextuality restricts the set of possible 
operations to compatible measurements, which in many cases need to be 
(approximately) projective or at least satisfy some analogous notion of 
repeatability and nondisturbance \cite{Guehne2010, Kujala2015}, in order to 
avoid the so-called ``compatibility loophole'' \cite{Guehne2010} or other 
similar classical explanations. Macrorealism has similar strong restrictions on 
the set of allowed measurements, namely, they must be noninvasive to avoid the 
clumsiness loophole or other forms of classical interpretation of the results 
\cite{WildeMizel2012}.

A strong motivation for developing such a notion of nonclassicality for single 
systems also arises from quantum information theory. Notions such as 
entanglement and nonlocality have been proved to play a role in quantum 
information tasks related to communication, such as, e.g., device-independent 
quantum key distribution \cite{AcinPRL2007}. That such notions should play a 
role also for tasks involving only single systems, such as, e.g., quantum 
computation, is less evident. Several recent results connected quantum 
contextuality with models of quantum computation such as, e.g., quantum 
computation via magic state injection or measurement based quantum computation 
\cite{Howard2014, DelfossePRX2015, RaussPRA2013, BermejoPRL2017, RaussPRA2017, 
AbramskyPRL2017, GalvaoPRA2017}. However, a natural question arises of whether 
this connection is fundamental or just related to the particular model used for 
quantum computation \cite{MarkiewiczPRA2014}. If one moves from compatible 
projective measurements to general instruments, it is no longer clear whether 
the notion of quantum contextuality make sense at all, due to the compatibility 
loophole mentioned above \cite{Guehne2010}.

In this paper, we go beyond such notions and introduce a notion of 
nonclassicality for the measurement statistics of a single system which is not 
restricted to specific measurement operations. The main tool of this 
investigation is the notion of memory cost of simulating temporal correlations. 
By temporal correlations we mean the observed statistics arising from sequences 
of measurements on a single system and memory roughly refers the amount of 
classical information that can be stored in the physical system.

The notion of memory cost has been explored in connection with classical simulations of quantum 
contextuality \cite{KleinmannNJP2011, Fagundes2017}, quantum simulation of classical 
stochastic processes~\cite{GarnerNJP2017} memory asymmetry between prediction and 
retrodiction~\cite{ThompsonPRX2018}, and in relation 
with the accuracy of classical and quantum clocks~\cite{Woods2018}. A related notion, i.e., 
that of communication cost, has been explored in relation to both Bell 
nonlocality \cite{Pironio_comm2003, Montina_comm2016} and temporal correlations 
\cite{BrierleyPRL2015, Zukowski2014}. Similar notions have been explored also 
in the prepare-and-measure scenario \cite{GallegoPRL2010, BrunnerPRL2013, 
DallarnoPRSA2017, DallarnoPRL2017, RossetPRX2018} and in connection with 
quantum information tasks such as random access codes \cite{BowlesPRA2015, 
TavakoliPRA2016, AguilarPRL2018, Miklin2019}.

In our approach, we go beyond the prepare-and-measure scenario by exploring 
arbitrary long sequences of measurements and we remove any restriction on the 
type of measurement by considering arbitrary quantum instruments. Our analysis 
is not only restricted to the differences between classical and quantum theory, 
but is extended to general probabilistic theories (GPTs) \cite{Ludwig:1985, 
Mittelstaedt:1998, Chiribella:2010PRA, Acin:2010PRL}, which embrace also the 
former theories. In particular, we derive inequalities on the observed 
probabilities that are able to discriminate between classical, quantum, and 
genuine GPT correlations. Moreover, as a further development of the ideas 
presented in Refs.~\cite{KleinmannNJP2011, Fagundes2017}, we show that in the 
framework of finite-state machines it is impossible to simulate contextual 
correlations on a qubit system, for a fixed initial state and arbitrary 
instruments.

The paper is organized as follows. In Sec.~\ref{sec:temp_corr}, we will 
introduce the basic notions and tools necessary for our analysis, namely, 
temporal correlations and the arrow of time polytope. In Sec.~\ref{sec:fin_st}, 
we will introduce finite-state machines in GPTs, in particular, also in 
classical and quantum theory. In Sec.~\ref{sec:temp_b}, we will discuss the 
existence of nontrivial temporal bounds for such theories and the impossibility 
of simulating contextual correlations on a qubit. Finally, we present the 
conclusions and an outlook of the paper.

\section{Temporal correlations}\label{sec:temp_corr}

We consider a box that accepts certain inputs from an input alphabet $\IA$ and 
produces outputs from an output alphabet $\OA$. The box is operated in a 
sequential fashion, see Fig.~\ref{fig:1}(a), such that, for instance, it first 
receives an input labeled by $x\in \IA$ yielding an output labeled by $a\in 
\OA$, subsequently it receives $y$ yielding $b$, and finally it receives $z$ 
yielding $c$. Prior to this sequence the box is initialized, such that its 
behavior is independent of anything except the input sequence $xyz$. 
Consequently, for a fixed input sequence $xyz\in \IA^3$, the admissible output 
sequences $abc\in \OA^3$ are governed by a probability distribution. If we now 
consider all possible inputs, we obtain the correlations $p(abc|xyz)$. Due to 
the time ordering of the inputs and outputs, these correlations must satisfy 
the arrow of time constraints \cite{ClementePRL2016},
\begin{gather}\label{eq:NS1}
\sum_c p(abc|xyz)= \sum_c p(abc|xyz')
\text{, for all } a,b\in \OA \text{ and all } x,y,z,z'\in \IA,
\\
\sum_{bc} p(abc|xyz)= \sum_{bc} p(abc|xy'z')
\text{, for all } a\in \OA \text{ and all } x,y,y',z,z'\in \IA.
\end{gather}
These constraints encode the fact that a future choice of an input, e.g., $z$ 
or $z'$ in Eq.~\eqref{eq:NS1}, must not influence previous outputs of the box, 
e.g., $a$ or $b$. This is in analogy to the nonsignaling conditions in the 
usual Bell scenario \cite{PopescuFPH1994}. The arrow of time constraints come 
solely from causality and hence, they must be satisfied not only in classical 
and quantum theory, but in any GPT.

We can represent the correlations $p(abc|xyz)$ as a vector with coordinates 
labeled by the possible sequences $abc$ and $xyz$. Due to the linearity of the 
arrow of time constraints, the set of correlations satisfying those forms a 
polytope. Its extremal points have been recently characterized 
\cite{AbbottPRA2016, HoffmannThesis2016, Hoffmann2018}. It is instructive to 
briefly sketch the central steps for the simple case of sequences of length 
three. All correlations in the corresponding polytope can be decomposed as
\begin{equation}\label{eq:decomp}
p(abc|xyz)= p(a|x) p(b|a;xy) p(c|ab;xyz),
\end{equation}
since the marginals on the right hand side are well defined (for the 
pathological cases where $p(ab|xy)=0$ we define the right hand side to be 
zero). Vice versa, taking valid probability distributions $p(a|x)$, 
$p(b|a;xy)$, $p(c|ab;xyz)$ over $a,b,c$, respectively, one always obtains an 
element of the polytope. Its extremal points are obtained by deterministic 
strategies, i.e., where each of the probability distributions on the right hand 
side of Eq.~\eqref{eq:decomp} consists only of probabilities $0$ or $1$. It 
easily follows that classical and quantum models can reach extremal points if 
enough memory is available. In more precise terms, each deterministic strategy 
can be reached if the box internally keeps a record of all previous inputs and 
outputs. Storing this record then requires the box to have memory. Of course, 
the notion of memory needs clarification, in particular if the box is described 
using quantum theory or a GPT, for details see Sec.~\ref{sec:fin_st}. Clearly, 
storing the full record of previous inputs and outputs is not necessarily 
memory optimal and gives rise to the question: What is the minimal number of 
states necessary to obtain certain correlations? How does such a number depend 
on the specific theory we use to describe the internals of the box?

\begin{figure}
\begin{tabular}{c@{\hspace{-2em}}r@{\hspace{2em}}c@{\hspace{-2em}}r}
\large{(a)}&&\large{(b)}&\\[-2em]
&\includegraphics[width=.41\textwidth]{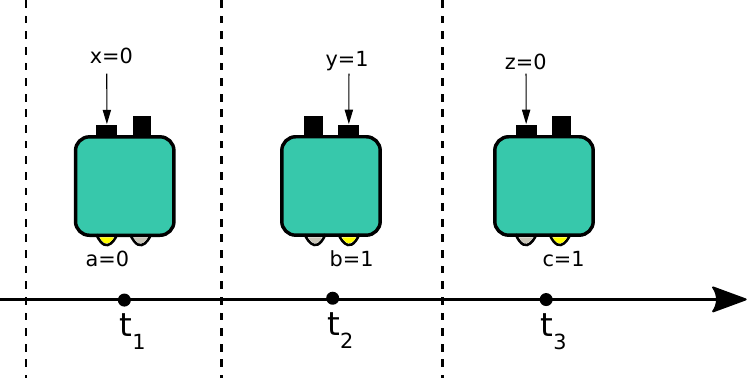}&
&\includegraphics[width=.41\textwidth]{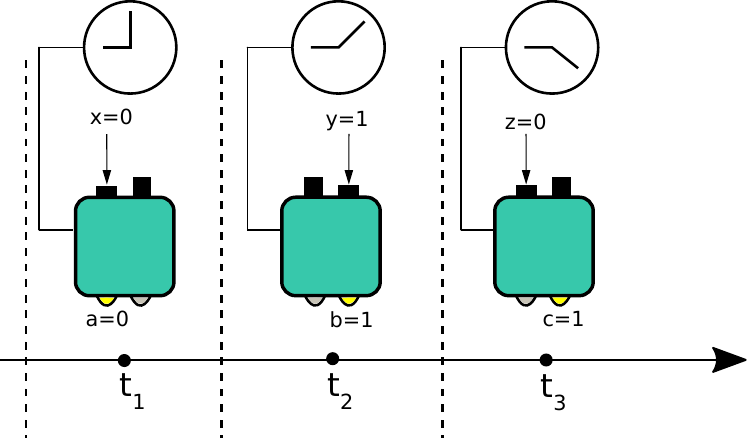}\\[1em]
\large{(c)}&&\large{(d)}&\\[-2em]
&\includegraphics[width=.48\textwidth]{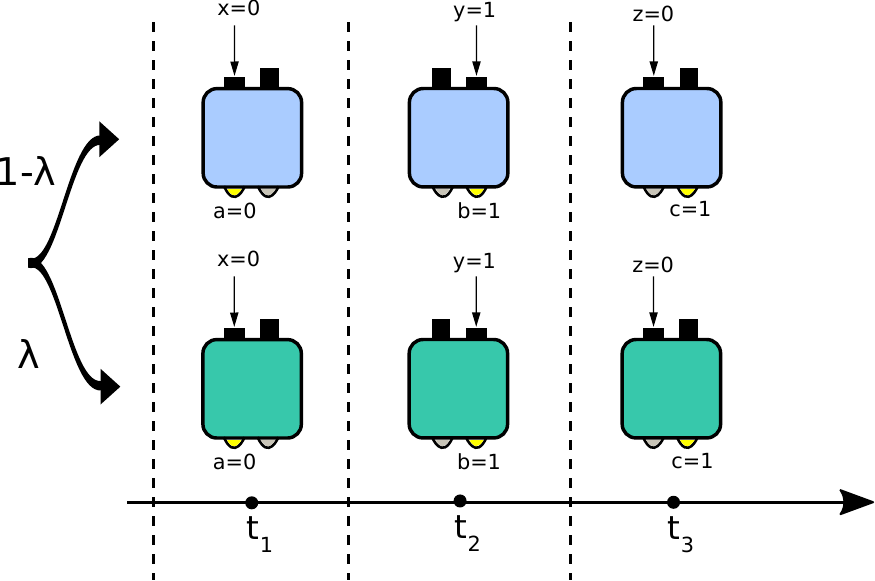}&
&\includegraphics[width=.48\textwidth]{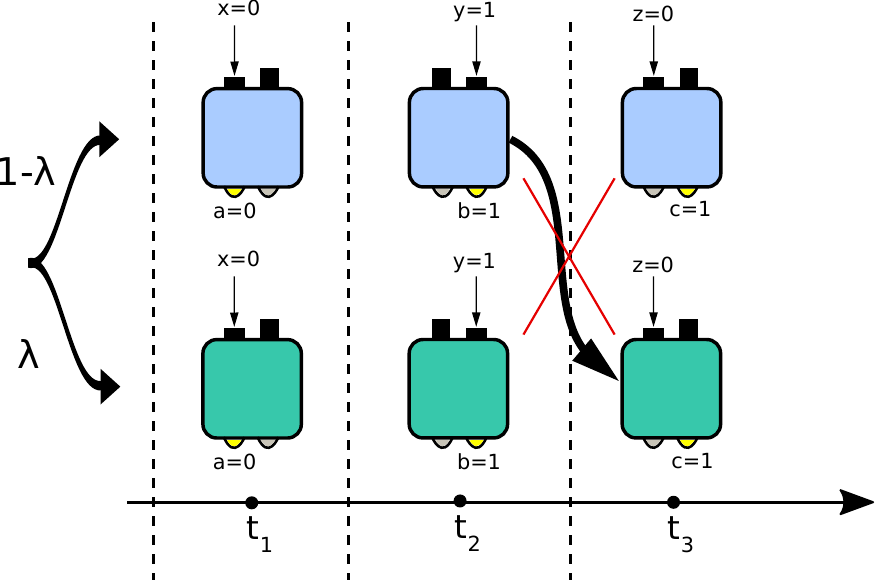}\\[1em]
\end{tabular}
\caption{\label{fig:1}%
Graphical representation of the possible situations. In each row, the three 
boxes represent the same physical device at different times/points in the 
sequence. (a) Main memory cost scenario. A single box with input sequence $xyz$ 
and output sequence $abc$. No external clock/memory is accessible to the box 
and hence its behavior is solely governed by its internal state. (b) 
Time-dependent operations. Additional time information is provided by a clock, 
which allows the box to perform time-dependent operations. This scenario is 
equivalent to the one used for discussing communication cost. (c) Allowed 
randomness. At the beginning of each run, the experimenter chooses with 
probability $\lambda$ the green box (bottom), or with probability $1-\lambda$ 
the blue box (top) and uses it to generate the whole sequence. (d) Forbidden 
randomness. It is not allowed to change the box at some point inside the 
sequence, since this would be a time-dependent operation.}
\end{figure}

An important element, in order to be able to speak about the memory cost of 
temporal correlations, is the requirement that all time-dependent information 
used to produce the outputs must be stored within the physical system used to 
implement the box. This implies that the physical operations performed to 
produce an output must be time-independent, e.g., the experimenter is not 
allowed to look at the wall clock and decide to implement in a different way 
the operation associated with a certain input $x$, as this will result in an 
additional source of memory, i.e., the clock keeping track of time. It is 
interesting to notice that the case where such time-dependent are admissible is 
equivalent to the case of quantum communication scenarios such as quantum 
random access codes or the scenario described by Brierley \emph{et 
al.} \cite{BrierleyPRL2015}. In fact, the latter scenario can be modeled as a 
network with ordered nodes, where a single physical system is transmitted 
through the nodes, and at each time step one of the nodes receives the system, 
performs a local operation, and transmit the system to the subsequent node. 
Since for each node it is known in advance in which part of the sequence it is 
situated, its local operations can be adapted to maximize a certain figure of 
merit defined in terms of probabilities of outcomes. This scenario covers the 
notion of ``communication cost'' and it must be distinguished from the notion 
of ``memory cost'' that is considered here. Moreover, even though in the memory 
cost scenario we are not allowed to change the operations throughout the 
sequence, it still makes sense to use classical randomness at the beginning of 
a sequence: at each experimental run, the experimenter can flip a coin and 
decide to perform the whole sequence of with one box or another. The resulting 
correlations will be a convex combination of the correlations obtained from 
either box. A graphical representation of the above ideas is presented in 
Fig.~\ref{fig:1}. These intuitive notions are made more rigorous in the next 
section.

\section{Finite-state machines}\label{sec:fin_st}

In this section, we formally define the classical, quantum, and GPT models for 
the box used in the previous section. In this model we assume that the box is 
implemented as a machine which acts on an internal state. Upon receiving an 
input $x$, the box operates on the internal state and produces the output $a$. 
The internal state is the specific model of the memory from the previous 
section. More precisely, we use the finite number of perfectly distinguishable 
states as a measure for the memory and for this reason we call this model a 
finite-state machine.

In a first step we need to describe the internal state $\omega$ and the 
operations $\II_{a|x}$ of the machine. We choose ordered vector spaces to 
describe the machine, which is an appropriate framework for a wide range of 
GPTs. In Appendix~\ref{app:gpt_def} we give a brief summary of this 
mathematical formalism. In brief, a GPT is then described by a real vector 
space $V$ with partial order ``$\le$'' and an order unit $e\in V$. In quantum 
theory $V$ would be the set of Hermitian operators, $A\le B$ would correspond 
to $B-A$ being positive semidefinite, and $e$ to the identity operator. 
Measurement outcomes are represented by effects $f\in V$ with $0\le f\le e$ and 
a measurement $\M_x$ is represented by a collection of effects $\M_x= 
(f_{a|x})_a$ with $\sum_a f_{a|x}= e$. The set of states $\ST$ is a subset of 
the dual space of $V$ such that the probability of outcome $a$ in the 
measurement $\M_x$ is given by $p(a|x)= \omega f_{a|x}$. Therefore $\omega e= 
1$ and $\omega f\ge 0$ for all $f\in V$ with $f\ge 0$. The operations 
$\II_{a|x}$ represent a specific way to implement a measurement, taking into 
account the change of the internal state $\omega$. More precisely, the linear 
map $\II_{a|x}\colon V\rightarrow V$ is such that $f_{a|x}= \II_{a|x}e$ is the 
effect describing the output $a$. In addition the positivity condition 
$\II_{a|x}f\ge 0$ for any $f\in V$ with $f\ge 0$ needs to be satisfied and 
further restrictions to $\II_{a|x}$ may apply depending on the specific GPT. If 
we group together the transformations $\II_x= (\II_{a|x})_{a\in \OA}$ for a 
fixed input $x$, then $\II_x$ is called an instrument. If we ignore the outcome 
$a$, then the instrument maps states to states, in the sense that $\omega 
\sum_a \II_{a|x}\in \ST$ for any state $\omega$.

Given the initial internal state $\omega$ of the finite-state machine and the 
instrument $\II_x= (\II_{a|x})_a$, the probabilities associated with a sequence 
of measurement are given by
\begin{equation}
p(a|x)= \omega \II_{a|x} e,\quad p(ab|xy)= \omega \II_{a|x} \II_{b|y}e, \quad 
\text{etc.}
\end{equation}
Note, that we write the transformations in the Heisenberg picture, so that the 
time ordering proceeds from the left to the right. For a general sequence of 
inputs $x_1x_2\cdots x_n= \vec x$ and outputs $a_1a_2\ldots a_n= \vec a$ we 
write
\begin{equation}\label{eq:def_gpt_prob}
p(\vec a|\vec x)\equiv p(a_1 \cdots a_n|x_1 \cdots x_n)
 = \omega\II_{a_1|x_1} \cdots \II_{a_n|x_n} e
 \equiv \omega \II_{\vec{a}|\vec{x}} e,
\end{equation}
We exemplify in the next sections how this expression is specialized to the 
classical and quantum case.

As we discussed previously, we exclude any external source of memory, such as a 
clock keeping track of time. This is formalized by the fact that all 
instruments solely depend on the input and in particular by the fact that all 
transformations are time-independent. In general, for a fixed GPT this 
requirement makes the set of achievable correlations nonconvex. Nevertheless, 
we can recover convexity by allowing the use of convex mixtures as follows. 
Before starting the experiment we use a random variable $\lambda$, distributed 
according to some probability distribution $q(\lambda)$, to decide which 
finite-state machine to use subsequently. Since the machine is characterized by 
the initial state $\omega_\lambda$ and the instruments $\II^\lambda_x$, this 
yields the correlations
\begin{equation}\label{eq:conv_lambda}
p(\vec{a}|\vec{x})= \sum_\lambda q(\lambda) \omega_\lambda
\II_{\vec{a}|\vec{x}}^\lambda e.
\end{equation}
The above procedure allows us to generate all correlations from the convex hull 
of correlations obtainable from a family of finite-state machines parametrized 
by $\lambda$.

Finally, we define the memory of the system using the GPT notion of capacity 
(cf. Ref.~\cite{MasanesNJP2011}), i.e., the size of the maximal set of 
perfectly distinguishable states. More precisely, we say that a GPT defines a 
$d$-state machine if $d$ is the maximal integer such that there exists a 
collection of $d$ states $(\omega_k)_k$ and $d$ effects $(f_k)_k$ such that
\begin{equation}\label{eq:def_dim}
\sum_k f_k \le e\text{ and } \omega_i f_j= \delta_{ij} \text{ for all }
 i,j.
\end{equation}
Namely, all effects are part of the same measurement, which is able to 
perfectly (i.e., probability one) discriminate among the states. This notion of 
capacity corresponds to the dimension of the Hilbert space in quantum mechanics 
and with the number of extremal points of the state simplex in classical 
probability theory (see, e.g., Ref.~\cite{MasanesNJP2011}).

It is instructive to discuss in more detail the classical and quantum case, 
which may be more familiar to the reader. We subsequently introduce a 
particular class of capacity-2 GPTs, the dichotomic norm cones 
\cite{KleinmannJPA2014}.

\subsection{Classical finite-state machines}\label{sec:cpa}

A classical finite-state machine \cite{PazBook2003} is described by its 
internal rules for state transitions and output probabilities. Given the 
classical state $\CC= \set{1,2,\dotsc, d}$, the observed probability 
distribution $p(\vec a|\vec x)$ for an input sequence $\vec x$ of length $n$ 
can be written as
\begin{equation}\label{eq:cmach}
 p(\vec a | \vec x)=
 \sum_{s_0,\ldots,s_n\in \CC}
 r(s_0)q(a_1,s_1|s_0,x_1) \cdots q(a_n,s_n|s_{n-1},x_n).
\end{equation}
Here, $r(s_0)$ describes the probability of preparing the initial state 
$s_0$ of the machine\footnote{Without loss of generality, we could assume a 
fixed pure initial state $s_0$, since we allow for convex mixtures of 
different machines. Nevertheless, we keep the notation with an initial 
distribution $r(s_0)$ over all pure states $\mathcal{C}$, i.e., a mixed state, 
to keep the analogy with the standard notation for GPT states ($\omega$) and 
quantum states ($\rho$).} and $q(a,s'|s,x)$ describes the 
probability that the machine yields the output $a$ and transition to 
the state $s'$, given that the internal state is $s$ and the input is $x$.  
As in Eq.~\eqref{eq:conv_lambda}, those machines can depend on a random variable $\lambda$
generated at the beginning of each sequence, i.e.,
\begin{equation}\label{eq:cmach_conv}
 p(\vec a | \vec x)=
 \sum_{s_0,\ldots,s_n\in \CC, \lambda} p(\lambda) r^\lambda(s_0)
 q^\lambda(a_1,s_1|s_0,x_1) \cdots q^\lambda (a_n,s_n|s_{n-1},x_n).
\end{equation}

For clarity reasons, we use only Eq.~\eqref{eq:cmach} in the following. The 
correlations $p(\vec a|\vec x)$ can be rewritten as
\begin{equation}
 p(\vec{a}|\vec{x})= \vtr\pi^\dag T(a_1|x_1)\cdots T(a_n|x_n)\vtr\eta
 \equiv \vtr\pi^\dag T(\vec{a}|\vec{x})\vtr\eta,
\end{equation}
where $\vtr\eta= (1,1,\ldots,1)^\dag$ is the $d$-dimensional vector of ones, 
$\vtr\pi$ is the vector representing the initial state, and $T(a|x)$ is the 
$d\times d$ transition matrix. Hence, $\pi_s= r(s)$ and 
$[T(a|x)]_{s,s'}= q(a,s'|s,x)$. The rules for probabilities that constrain 
$q(a,s'|s,x)$ translate to $[T(a|x)]_{s,s'}\ge 0$ for all $s,s',a,x$, and 
$\sum_a [T(a|x)\vtr\eta]_s=1$ for all $s,x$.

Translating the above in the languages of GPTs, we let $V= \reals^d$ and set 
the order unit $e$ to $\vtr \eta$. The partial order is such that $\vtr v\le 
\vtr w$ if $v_s\le w_s$ for all $s$. Then the set of states is given by by the 
canonical $(d-1)$-dimensional simplex,
\begin{equation}
 \ST= \set{ \vtr \varpi \in \reals^d | \vtr\varpi \ge 0 \text{ and } 
 \vtr\varpi^\dag \vtr\eta=1}.
\end{equation}
In particular $\vtr\pi$ is a state. Analogously, the transition matrix $T(a|x)$ 
corresponds to the instruments $\II_{a|x}$, whereas the effects can be obtained 
as $f_{a|x}:= T(a|x)\vtr\eta$. It can be easily seen that $d$ correspond 
exactly to the capacity defined according to Eq.~\eqref{eq:def_dim}.

\subsubsection{Classical finite-state machines and Leggett-Garg's macrorealist models}
It is interesting at this point to briefly compare the model in Eq.~\eqref{eq:cmach_conv} with the macrorealist model of Leggett and Garg~\cite{LeggettPRL1985}. A macrorealist model can be simply obtained by reducing the set of possible internal states to a single one, i.e., $d=1$, and re-introducing the time-dependence of operations.
\begin{equation}\label{eq:MR}
 p(\vec a | \vec x)=
 \sum_{\lambda} p(\lambda) q^\lambda_{t_1}(a_1|x_1) \cdots q^\lambda_{t_n} (a_n|x_n),
\end{equation}
where the dependency on $s_0,\ldots,s_n$ becomes trivial and is then removed. We recall that macrorealist models are based on two assumptions: {\it macrorealism per se}, i.e., the existence of a classical probability, and {\it noninvasive measurability}, i.e., the assumption that the measurement has no effect on the subsequent evolution of the system. The finite-state machine model can be seen as arising from the macrorealist model via a relaxation of the assumption of a noninvasive measurement: the measurement can be invasive up to a certain amount quantified by the internal memory of the system, e.g., for a two state-machine the measurement can encode at most one bit of information in the system. Notice that, however, usually Leggett-Garg assumptions allow the operations to be time-dependent. 

It is interesting to remark that similar ideas have been already employed in Leggett-Garg tests to tighten the clumsiness loophole. Under the assumption of a classical model with two internal states, Knee {\it et al.} \cite{Knee2016} were able to quantify the measurement invasivity via a control experiment, and consequently modify the classical bound for the Leggett-Garg inequality. In agreement with our argument above, the work of Knee {\it et al.} shows how the notion of finite memory can be used as a relaxation of the assumption of a noninvasive measurement.

\subsection{Quantum finite-state machines}

The quantum case is perhaps the most familiar to readers from quantum 
information. The probability distribution is obtained by sequences of 
generalized measurements $\M_x= (E_{a|x})_a$ on a single system described by a 
Hilbert space of fixed dimension $d$. The outcomes of the measurement are 
described by positive semidefinite operators $E_{a|x}\ge 0$ with $\sum_a 
E_{a|x}= \openone$.

In order to discuss sequential measurements, however, we need to know the 
post-measurement state, or, better, the transformation induced by the 
measurements. This information is provided by a quantum instrument $\II_x$, 
defined as a collection of completely positive maps $\II_x= (\II_{a|x})_a$, 
from the space of linear operators into itself, that sum up to a unital map, 
i.e., $\sum_a \II_{a|x}(\openone)= \openone$, corresponding to the rule of 
preservation of probability in the Heisenberg picture, see, e.g., 
\cite{HeinosaariZiman2011}. Each instrument defines a generalized measurement 
through the formula $E_{a|x}= \II_{a|x}(\openone)$. Similarly to the previous 
cases, we can shorten the notation by defining $\II_{\vec{a}|\vec{x}} := 
\II_{a_1|x_1} \circ \ldots \circ \II_{a_n|x_n},$ where $\circ$ denotes the 
composition of maps and write
\begin{equation}\label{eq:def_qpa_prob}
p(\vec{a}|\vec{x})= \tr[\rho\,\II_{\vec{a}|\vec{x}}(\openone) ].
\end{equation}

As mentioned before, quantum theory is a particular case of a GPT, where the 
vectors space $V$ is the set of Hermitian operators, the partial order is 
defined through positive semidefiniteness and the order unit $e$ is given by 
$\openone$. The set of states is given by the density operators, identified by 
the Hilbert--Schmidt inner product with the elements of the dual space of $V$,
\begin{equation}
 \ST= \set{ X \mapsto \tr(\rho\,X) | \rho\ge 0 \text{ and } \tr(\rho)=1 }.
\end{equation}
Hence Eq.~\eqref{eq:def_qpa_prob} and Eq.~\eqref{eq:def_gpt_prob} are 
equivalent. It is then clear that the capacity of the system, defined as the 
number of perfectly distinguishable state \cite{FritzNJP2010, Hoffmann2018} 
precisely corresponds to the dimension of the Hilbert space. It is important to 
remark that we need to consider the general formalism of quantum instruments, 
since if the measurement devices would merely act projectively, there would be 
nontrivial limitations on the achievable correlations that are valid for 
arbitrary dimensions \cite{BudroniPRL2013, BudroniPRL2014}.

\subsection{GPT two-state machines}\label{sec:gpa}

We already provided a definition of GPT finite-state machines at the beginning 
of Sec.~\ref{sec:fin_st}. In this section, we specialize this definition by 
considering a class GPTs where the effects belong to a dichotomic norm cone. 
These theories are a generalization of the classical bit (cbit) and quantum bit 
(qubit), in the sense that they have capacity two, i.e., they allow for a set 
of perfectly distinguishable states, in the sense of Eq.~\eqref{eq:def_dim}, of 
at most size two. We then specialize our discussion to the case of hyperbits 
(hbits) \cite{PawlowskiPRA2012} and generalized bits (gbits) 
\cite{BarrettPRA2007}. The former are a generalization of the Bloch sphere to 
dimension higher than three, whereas the latter are the local part of a 
Popescu--Rohrlich box \cite{PopescuFPH1994}. We also provide a more detailed 
discussion of GPTs in Appendix~\ref{app:gpt_def}.

Consider the vector space $V:= \reals\times \reals^n$, and the partial order 
where $(t,\vtr x)\ge 0$ if $t\ge \norm{\vtr x}$. Here, $\norm{\vtr x}$ is any 
norm in $\reals^n$. We define the order unit $e:= (1, \vtr 0)$. This implies 
that effects are vectors $f= (t,\vtr x)$ such that $\norm{x}\le \min\set{t, 
1-t}$. The states for a dichotomic norm cone are the maps $\omega\colon (t,\vtr 
x)\mapsto t+\vtr w^\dag \vtr x$ with the condition $\norm{\vtr w}_*\le 1$, 
where ${\norm{\vtr w}_* := \sup\set{\vtr w^\dag \vtr y | \norm{\vtr y}\le 1}}$ 
is the dual norm of $\norm{\,\cdot\,}$. A peculiarity of this GPT is that it 
has exactly capacity two, independent of $n$ or the choice of the norm 
$\norm{\,\cdot\,}$. We provide a proof of this fact in Appendix~\ref{app:dim}.

Depending on the norm chosen and on $n$ we have different GPTs. If we take 
$\norm{\vtr x}$ to be the Euclidean (or $\ell_2$) norm, i.e., $\norm{\vtr x}^2= 
\sum_i x_i^2$, we obtain hbits, and specifically cbits for $n=1$, qubits for 
$n=3$ and more general hbits for $n>3$. If we take $n=2$ and the Manhattan (or 
$\ell_1$) norm, i.e., $\norm{\vtr x}= \sum_i |x_i|$, we obtain a gbit. For the 
case of the Euclidean norm, the dual norm is also the Euclidean norm itself, 
whereas the dual of the Manhattan norm is the supremum (or $\ell_\infty$) norm, 
i.e., $\norm{\vtr w}_*= \max_i |w_i|$.

\section{Bounds on temporal correlations}\label{sec:temp_b}

In this section, we consider the simplest nontrivial scenario, a sequence of 
two measurements, with inputs $x,y$ and outputs $a,b$, with $a,b,x,y=0,1$. We 
are interested in bounds on the  sum of correlations
\begin{equation}
\soc= p(01|00) + p(10|10) + p(10|11).
\end{equation}
Similar expressions have been considered in Ref.~\cite{HoffmannThesis2016, 
Hoffmann2018, Spee2018}. Clearly, the trivial bound $\soc\le 3$ holds. For 
hbits the value $\soc=3$ cannot be reached and therefore there must exist a 
nontrivial bound $\soc\le \Omega_{\mathrm{hbit},n}$ for any dimension $n$ of 
the hbit, in particular for the cbit ($n=1$) and the qubit ($n=3$). A simple 
analytical proof of $\Omega_{\mathrm{hbit},n}<3$ is presented in 
Appendix~\ref{app:bound}.

\subsection{Measure-and-prepare strategies}

The analysis of the case of sequences of length two can be greatly simplified 
using measure-and-prepare instruments. These are instruments of the form 
$\TT_x= (f_{a|x}\sigma_{a|x})_a$, where $\M_x= (f_{a|x})_a$ is a measurement 
and $(\sigma_{a|x})_a$ is a collection of states. Hence $\TT_x$ can be 
implemented by first measuring $\M_x$ and then, depending on the outcome $a$, 
preparing the state $\sigma_{a|x}$.

Now, for a sequence of length two, the correlations are given by
\begin{equation}
 p(ab|xy)= \sum_\lambda p(\lambda)\omega^\lambda \II^\lambda_{a|x}f_{b|y}^\lambda,
\end{equation}
where $\omega^\lambda$ is given by the initialization procedure of the 
individual finite-state machines participating in the mixture of machines. 
Clearly, the extremal values $\soc$ can be achieved by a single finite-state 
machine and hence in the following we will omit the index $\lambda$ and the 
summation of $\lambda$.

The instruments $\II_x$ can be replaced by measure-and-prepare instruments, by 
letting $f_{a|x}= \II_{a|x}e$ and $\sigma_{a|x}= \omega 
\II_{a|x}/\omega(f_{a|x})$ if the denominator is nonzero, or $\sigma_{a|x}= 
\omega$. Then $p(ab|xy)= \omega f_{a|x}\sigma_{a|x} f_{b|y}$. Hence we can 
equivalently replace $\II_{a|x}$ by the prepare-and-measure strategy $\mathcal 
T_{a|x}= f_{a|x}\sigma_{a|x}$. Using this simplification, we obtain
\begin{equation}\label{eq:Smp}\begin{split}
 \soc &= p(0|0) + p(1|1) + p(10|10) - p(00|00) - p(11|11)\\
      &= p(0|0)[1- p(0|0;00)] + p(1|1)[1+p(0|1;10)-p(1|1;11)]\\
      &= \omega(f_{0|0}) [1- \sigma_{0|0}(f_{0|0})]
      + \omega(f_{1|1}) [1+ \sigma_{1|1}(f_{0|0}-f_{1|1})],
\end{split}
\end{equation}
where we used the notation $p(b|a;xy)$ for the probabilities conditioned on 
previous outputs.

\subsection{Analytical and numerical bounds}

Since $\soc=3$ cannot be reached with hbits, there must be a finite gap between 
the actual bound for cbits, qubits, and hbits with a Bloch sphere of fixed 
dimension. In fact, the sets of states and effects are compact, and the 
expression $\soc$ can be written as a continuous function from the set of 
states and effects into the interval $[0,3]$, so its image must be compact. In 
this section, we explore in more detail the bounds for cbits, qubits, and hbits 
via numerical methods.

\subsubsection{Classical bit}

For the cbit case, we  use the representation from Sec.~\ref{sec:cpa}, 
specifically, $\omega$ is represented by $(1,0)$, $\sigma_{i|i}$ by 
$(s_i,1-s_i)$, and $f_{i|i}$ by $(a_i,b_i)^\dag$, where $s_i, a_i, b_i\in 
[0,1]$. Then Eq.~\eqref{eq:Smp} reads
\begin{equation}
 \soc= a_0[1-s_0a_0-(1-s_0)b_0]+a_1[1+s_1(a_0-a_1)+(1-s_1)(b_0-b_1)].
\end{equation}
Only $a_0$ and $a_1$ appear nonlinearly in this expression. Therefore, the 
maximum of $\soc$ is attained when all remaining parameters are either $0$ or 
$1$. This leaves us with a two-dimensional, at most quadratic optimization, 
which can be performed at once. For the maximal value $\Omega_\mathrm{cbit}$ 
of $\soc$ using classical bits we then obtain
\begin{equation}
 \Omega_\mathrm{cbit}= \frac94.
\end{equation}
This maximum occurs at a unique point, where $s_1= b_1=0$, $b_0= s_0= a_1=1$, 
and $a_0= \frac12$. Hence, an optimal machine is given by the initial state 
$\vtr \pi^\dag= (1,0)$ and the transition matrices
\begin{equation}
T(0|0)= \begin{pmatrix}
 \frac12 & 0 \\
 1 & 0 \\
\end{pmatrix},\quad
T(1|0)= \begin{pmatrix}
 \frac12 & 0 \\
 0 & 0 \\
\end{pmatrix},\quad
T(1|1)= \begin{pmatrix}
 0 & 1 \\
 0 & 0 \\
\end{pmatrix}= T(0|1)^\dag.
\end{equation}
Note, that while the solution for the chosen parametrization is unique, the 
transition matrices are not unique.

\subsubsection{Quantum bit}

For the qubit case, we can proceed similarly to Ref.~\cite{Hoffmann2018}. First 
we note that in Eq.~\eqref{eq:Smp}, the initial state $\omega$ can be replaced 
by a pure state, so that $\omega\colon X\mapsto \braket{0|X|0}$. The expression 
$\soc$ can then be written as
\begin{equation}
\soc= \braket{0| E_{0|0} | 0} \left[1 - \tr[\sigma_{0} E_{0|0} ] \right ]+ 
\braket{0| E_{1|1} | 0} \left[1 + \tr[\sigma_{1}(E_{0|0}- E_{1|1}) ] \right ], 
\end{equation}
where $0\le E_{i|i}\le \openone$ are effects and $\sigma_0$ and $\sigma_1$ are 
density operators. Since the latter occur only linearly in $\soc$, we can 
substitute them with pure states $\ket{\psi_0}$ and $\ket{\psi_1}$, 
respectively. The maximum of $\soc$ for qubits is hence given by
\begin{equation}\label{eq:b3qubit}
 \Omega_\mathrm{qubit}=
 \max_{\stackrel{\ket{\psi_0}, \ket{\psi_1}}{E_{0|0}, E_{1|1}}} \left[
  \braket{0| E_{0|0} | 0} \left(1-\braket{\psi_{0}| E_{0|0}| \psi_{0}}\right)+
  \braket{0| E_{1|1} | 0} \left(1+\braket{\psi_{1}|E_{0|0}-E_{1|1}|\psi_{1}} 
\right) \right],
\end{equation}

By parametrizing $E_{0|0}$, $E_{1|1}$, $\ket{\psi_{0}},\ket{\psi_{1}}$ with 
real parameters, one can write the expression in Eq.~\eqref{eq:b3qubit} as 
fourth degree polynomial. This can be further simplified, by taking $E_{0|0}$ 
$E_{1|1}$, $\ket{\psi_{0}},\ket{\psi_{1}}$ as real expression, which lowers the 
number of parameters to ten.\footnote{Since the upper bound is calculated by 
polynomial optimization methods, it is more convenient to keep the expression 
and constraints in polynomial form, rather than minimizing the number of 
variables. For example, a parametrization of a pure state as ${\cos \theta 
\ket{0} + \sin \theta \ket{1}}$ removes one variable and one constraint, but it 
is no long a polynomial in the parameters.} The reduction to the real part of a 
qubit does not affect the optimality as we show in the next section, see 
Eq.~\eqref{eq:b3_simpl}.

It is always possible to obtain a lower bound 
$\Omega^\mathrm{feas}_\mathrm{qubit}$ on $\Omega_\mathrm{qubit}$ by guessing 
appropriate values for the free parameters. An upper bound, 
$\Omega^\mathrm{Lass}_\mathrm{qubit}$, can be obtained via Lasserre's method 
\cite{Lasserre2001} of polynomial optimization based on moment matrices and 
semidefinite programming (SDP) \cite{SDP_review}, which provides analytical 
upper bounds up to the numerical precision. That is,
\begin{equation}
\Omega_{\rm qubit}^{\rm feas} \le \Omega_{\rm qubit} \le \Omega_{\rm 
qubit}^{\rm Lass}.
\end{equation}
With the simplifications used above, the upper and lower bounds coincide up to 
the numerical precision of $10^{-5}$. We have,
\begin{equation}\label{eq:qubit_bound}
\Omega_{\rm qubit}^{\rm feas} \approx \Omega_{\rm qubit}^{\rm Lass} \approx 
2.35570,
\end{equation}
showing a gap between the cbit and qubit case. A feasible solution is given by 
the post-measurement states and effects,
\begin{equation}
 \ket{\psi_{0}}\approx 0.408\ket0-0.913\ket1,
 \quad
 \ket{\psi_{1}}\approx 0.640\ket0+0.768\ket1,
\end{equation}
and the effects
\begin{equation}
 E_{0|0}= \openone-\ketbra{\psi_{0}}, \text{ and }
 E_{1|1}= \ketbra\phi, \text{ where }
 \ket\phi\approx 0.971\ket0-0.238\ket1.
\end{equation}

\subsubsection{Hyperbit}

For the case of hbits, and also the more general dichotomic norm cones, we use 
the parametrization $\omega\colon (t,\vtr x)\mapsto t+\vtr w^\dag \vtr x$ and 
$\sigma_{i|i}\colon (t,\vtr x)\mapsto t+\vtr w_i^\dag \vtr x$ for the states 
and $f_{i|i}= (t_{i},\vtr{f}_{i})$ for the effects. Then Eq.~\eqref{eq:Smp} 
reads
\begin{equation}
 \soc=
   (t_0+\vtr w^\dag\vtr f_0) [ 1-t_0-\vtr w_0^\dag \vtr f_0 ]
 + (t_1+\vtr w^\dag\vtr f_1) [ 1+t_0-t_1 + \vtr w_1^\dag (\vtr f_0-\vtr f_1 )].
\end{equation}
When maximizing $\soc$, we can eliminate the maximization over $\vtr{w}_0$ and 
$\vtr{w}_1$, by choosing appropriate vectors with $\norm{\vtr{w}_i}_*=1$ such 
that $\vtr w_0^\dag \vtr f_0= \norm{\vtr f_0}$ and $\vtr w_1^\dag(\vtr f_0-\vtr 
f_1)= \norm{\vtr f_0- \vtr f_1}$. The maximal value of $\soc$ for a given 
dichotomic norm cone is hence
\begin{equation}\label{eq:b3_simpl}
 \Omega_\mathrm{dnc}= \max_{\substack{\vtr w, t_0, \vtr f_0\\ t_1, \vtr f_1}}
 \left\{
   (t_0 + \vtr w^\dag \vtr f_0) [ 1-t_0+\norm{\vtr{f}_0}\, ]
 + (t_1 + \vtr w^\dag \vtr f_1) [ 1+t_0-t_1+\norm{\vtr{f}_0 - \vtr{f}_1}\,]
 \right\},
\end{equation}
where the constraints of the optimization are $\norm{\vtr w}_*\leq 1$ and $\norm{\vtr f_i} \leq \min\{t_i, 1-t_i\}$. 
For the case of hbits, both $\norm{\cdot }_*$ and $\norm{ \cdot}$ correspond to the $\ell_2$ norm , hence the conditions are invariant  under 
orthogonal transformations as it is the case for the function to be maximized, which depends only on the norm of $\vtr f_i$ and the scalar products between $\vtr w$ and $\vtr f_i$. Since the only contribution for $\vtr w$ comes from 
the component in the span of $\vtr{f}_0, \vtr{f}_1$, the problem reduces to a 
two-dimensional one. 
This is equivalent to the qubit case with the Bloch ball 
restricted to the $xz$-plane, both for states and effects. This implies that 
the bound for hbits coincide with the bound for qubits. We thus have
\begin{equation}
\Omega_{\rm hbit} \approx 2.35570,
\end{equation}
as in Eq.~\eqref{eq:qubit_bound}.

\subsubsection{Generalized bit}

The case of gbits differs from the previous one because we can actually reach 
$\soc=3$ already for a two-state machine, namely the dichotomic norm cone with 
$n=2$ and the $\ell_1$ norm. This model corresponds to the local part of a 
Popescu--Rohrlich box \cite{PopescuFPH1994,BarrettPRA2007}. The space of 
effects is a polytope with extremal effects given by the extremal point of the 
two-dimensional $\ell_1$ norm, i.e., $a_{\pm i}= \frac12 (1,\pm \vtr e_i)$, 
with $\vtr e_i$ the canonical vectors in $\reals^2$. Then, the states are the 
$\omega= (1,\vtr{w})$ with $\vtr{w}$ in the square $[-1,1]\times [-1,1]$, i.e., 
the unit ball with respect to the $\ell_\infty$ norm. The choices
\begin{equation}
 \vtr w^\dag= (1,-1),\quad
 \vtr w_0^\dag= (-1,1),\quad
 \vtr w_1^\dag= (1,1)
\end{equation}
and
\begin{equation}
 \vtr f_0= \vtr e_1,\quad \vtr f_1= -\vtr e_2
\end{equation}
yield, according to Eq.~\eqref{eq:b3_simpl}, the algebraic maximum for $\soc$, 
i.e., $\soc=3$. We thus have
\begin{equation}
\Omega_{\rm gbit}=3
\end{equation}
for gbits and hence also for the set of all dichotomic norm cones with the same norm and arbitrary $n$.

\subsection{Impossibility of simulating contextual correlations with general 
instruments on a qubit}

In this section, we investigate whether qubit machines are able to simulate 
some contextual correlations that arise in higher dimensional quantum systems. 
In Ref.~\cite{KleinmannNJP2011} it was proved that in order to simulate all 
deterministic predictions associated with the observables of the
Peres--Mermin square \cite{Peres1990,Mermin1990}, a
classical machine with at least $4$ states is necessary. This result was obtained in 
the framework of tests of contextuality involving sequential measurements 
\cite{Guehne2010}, in which the relevant compatibility notion is given by the 
nondisturbance among compatible measurements and repeatability of outcomes, 
e.g., if $\M_x$ and $\M_y$ are compatible measurements in the measurement 
sequence $\M_x\M_y\M_x$, the outcome for the first measurement of $\M_x$ will 
be repeated in the second measurement of $\M_x$.

We derive here a related result by showing that even a qubit is not sufficient 
to exhibit contextual correlations. For this we use a rather broad notion of 
contextuality. Consider a box with inputs from an alphabet $\IA$ and outputs 
from an alphabet $\OA$ as before. The input sequences are restricted such that 
a sequence $\vec x$ is admissible if and only if all inputs are from the same 
context $\CA\subset\IA$, i.e., $\set{x_i|i}\subset \CA$. A context $\CA$ is a 
set of inputs, such that $p(\vec a|\vec x)= p[\pi(\vec a)|\pi(\vec x)]$ for any 
inputs sequence $\vec x$ from $\CA$, any output sequence $\vec a$, and any 
permutation $\pi$. In addition we assume that any input is repeatable, i.e., $ 
p(\vec ab|\vec xx_i)= p(\vec a|\vec x)\delta_{b,a_i}$ for any position $i$ in 
any admissible sequence.

Such a box is noncontextual, if all correlations of the box (using only 
admissible input sequences) can be reproduced by a box without memory, i.e., by 
a noncontextual model. We claim that any such box implemented on a qubit is 
noncontextual.

We start the proof of this statement by determining those inputs, which cannot 
require the use of memory. First, if an input $z$  ever produces only  the output $c$, within all 
admissible input sequences, then we can eliminate this 
input from our considerations. This is the case, because in any sequence we can 
permute $z$ to the end of the sequence. Then
\begin{equation}
 p(\vec ac|\vec xz)= p(\vec a|\vec x)p(c|\vec a;\vec xz)= p(\vec a|\vec x),
\end{equation}
where the first equality is due to Eq.~\eqref{eq:decomp} and the second due to 
the assumption that only  the output $c$ ever occurs. Second, assume that for a 
certain input $z$, whenever it occurs in an admissible sequence, the internal 
state of the machine before the input $z$ is only ever the state $\rho$. Again 
we can eliminate this input from our considerations, because the output for $z$ 
and the state after the output can be determined without considering the state. 
Third, we can ignore the pathological cases of inputs, which are not member of 
any context. In the following we assume without loss of generality, that the 
box does not have any input falling under the those three cases just discussed.

Next, we show that for any input $z$ the instrument $(\II_{c|z})_c$ must be a 
measure-and-prepare instrument of the form
\begin{equation}\label{eq:qubit_mnp}
 \II_{c|z}\colon X\mapsto \ketbra{\psi_{c,z}} X\ketbra{\psi_{c,z}} \text{ with 
 }
 \braket{\psi_{c,z}|\psi_{c,z}}\in \set{0,1}.
\end{equation}
This can be seen as follows. According to the assumptions, there are two input 
sequences $\vec xz$ and $\vec yz$ and corresponding output sequences $\vec ac$ 
and $\vec bc$, so that the state before the input $z$ is $\rho$ and $\rho'$, 
respectively, with $\rho\ne \rho'$. Using Eq.~\eqref{eq:decomp} and 
Eq.~\eqref{eq:def_qpa_prob} we have
\begin{align}
 p(\vec ac|\vec xz)\delta_{c,c'}=
 p(\vec acc'|\vec xzz)&= p(\vec a|\vec x)p(cc'|\vec a;\vec xzz)
   = p(\vec a|\vec x) \tr[\rho\phantom{'}\,\II_{c|z}\II_{c'|z}\openone]
 \text{ and}\\
 p(\vec bc|\vec yz)\delta_{c,c'}=
 p(\vec bcc'|\vec yzz)&= p(\vec b|\vec y)p(cc'|\vec b;\vec yzz)
   = p(\vec b|\vec y) \tr[\rho'\,\II_{c|z}\II_{c'|z}\openone],
\end{align}
where $p(\vec a|\vec x)>0$ and $p(\vec b|\vec y)>0$. Therefore for $c\ne c'$,
\begin{equation}
 \tr[\bar\rho\,\II_{c|z}\II_{c'|z}\openone]= 0
\end{equation}
with $\bar\rho= (\rho+\rho')/2$. Since $\rho\ne \rho'$ and we assume a qubit 
system, the mixture $\bar\rho$ has necessarily rank two, i.e., $\bar\rho \ge 
\epsilon\openone$ for some $\epsilon>0$. We arrive at the condition
\begin{equation}
 \sum_{i,j} \tr[K_i Q_j Q_j^\dag K_i^\dag]=0,
\end{equation}
where $K_i$ and $Q_j$ are the Kraus operators associated, respectively, with 
the instruments $\II_{c'|z}$ and $\II_{c|z}$, e.g., $\II_{c'|z}X= \sum_j 
K_j^\dag X K_j$. Then $K_i Q_j=0$ for all $i,j$. Similarly, exchanging $c$ with 
$c'$, we obtain $ Q_j K_i=0 $ for all $i,j$. This implies that $K_i$ and $Q_j$ 
are of rank one and that $K_i$ is proportional to $K_{i'}$ as well as $Q_j$ 
being proportional to $Q_{j'}$, for all $i,i'$ and $j,j'$. Hence we can omit 
the indices $i,j$ and consider simply $K$ and $Q$. Note that from $\sum_c 
\II_{c|z}\openone= \openone$, the condition $Q^\dag Q\le \openone$ follows 
which allows us to write $Q=\KetBra\alpha\beta$ with $\braket{\alpha|\alpha}=1$ 
and $\braket{\beta|\beta}\le 1$. Now, for $c=c'$ we obtain
\begin{equation}
\tr(\bar\rho \II_{c|z}\II_{c|z}\openone)=
\tr(\bar\rho \II_{c|z}\openone),
\end{equation}
which implies $(Q^\dag)^2Q^2= Q^\dag Q$. It follows that either $\ket\beta =0$ 
or $\ket\alpha$ and $\ket\beta$ are equal up to a phase and hence $\II_{c|z}$ 
is as stated in Eq.~\eqref{eq:qubit_mnp}.

As final step we need to show that there is no contextuality for projective 
qubit instruments. Given an admissible input sequence $\vec xyz$, and an output 
sequence $\vec abc$ such that $p(\vec ab|\vec xy)>0$, we have
\begin{equation}
 p(\vec abc|\vec xyz)
 = p(\vec ab|\vec xy) \abs{\braket{\psi_{b,y}|\psi_{c,z}}}^2
 \text{ and } p(\vec abcb|\vec xyzy)
 = p(\vec ab|\vec xy) \abs{\braket{\psi_{b,y}|\psi_{c,z}}}^4.
\end{equation}
The left hand side of both expressions has to be equal, yielding 
$\abs{\braket{\psi_{b,y}|\psi_{c,z}}}\in \set{0,1}$.

Consequently, any two inputs within a context are realized by the same 
projective instrument, except for some relabeling of the outcomes. We choose a 
specific measurement within one context, say $y$, so that $\II_{a|x}=\sum_b 
\II_{b|y}f^b(a|x)$ with some coefficients $f^b(a|x)\in \set{0,1}$. This way we 
can write for any correlations of this context
\begin{equation}
 p(\vec a|\vec x)=\sum_b p(b|y)\prod_i f^b(a_i|x_i),
\end{equation}
which is exactly the formula for a one-state machine, i.e., a noncontextual 
model.

This concludes the proof of our statement, due to the following
observation. If two contexts share an observable, then our argument
already applies and the union of both contexts must admit a
noncontextual model and hence the union of both contexts is again a
context. Eventually, we can join contexts until all contexts are
mutually disjoint. For each disjoint set we can construct a
noncontextual model, and since there are no admissible sequence
involving two different contexts, we have constructed a
noncontextual model for all admissible input sequences.

\section{Conclusions and outlook}

We introduced the notion memory cost of simulating temporal correlations based 
on the notion of finite-state machine, i.e., a physical system accepting an 
input at each time instant and generating an outcome and an internal state 
transition according to probabilistic rules. We investigated the correlations 
obtainable via such finite-state machines operating according to different 
probability theories, i.e., classical, quantum, or GPT. Our framework allow us 
to derive inequalities able to discriminate among different theories for the 
simplest nontrivial case, i.e., two-state machines, two inputs, two outputs, 
and sequences of length two. Moreover, we investigated, from the perspective of 
quantum finite-state machines, the possibility of simulating contextual 
correlations with a qubit and answered this question in the negative.

Our framework provides a notion of nonclassicality for single systems, which is 
based solely on observed correlations and does not make any assumption of the 
type of measurements involved, e.g., compatibility or noninvasiveness. We 
believe that several problems in quantum foundations and quantum information 
could be studied in this framework. For instance, a notion of nonclassicality 
for single systems, i.e., quantum contextuality, has recently been suggested as 
a resource for quantum computation. On the other hand, memory has been 
identified as a resource needed to simulate contextual correlations classically 
\cite{KleinmannNJP2011, Fagundes2017}. In addition, a different notion of 
contextuality for sequential operations has been defined and connected to 
speed-up in quantum computation \cite{Mansfield2018}. Our work could provide a 
general framework to discuss such different results and understand better the 
connection between memory cost of (classical) simulations, contextual 
correlations, and advantages in computation. Moreover, the idea of computation 
in GPTs, such as Spekkens' toy model \cite{SpekkesToy07}, that are intermediate 
between classical and quantum probability has been recently investigated 
\cite{JohanssonQIP2017, Johansson2017}. In particular, this GPT can be exactly 
simulated with two classical bits.

\acknowledgments

The authors would like to thank
Rafael Chaves,
Andrew Garner,
Otfried G\"uhne,
Jannik Hoffmann,
Niklas Johansson,
Jan-{\AA}ke Larsson,
Nikolai Miklin,
Miguel Navascu\'es,
Cornelia Spee,
Giuseppe Vitagliano, and
Mischa Woods for discussions.
This work has been supported by the
Austrian Science Fund (FWF): M 2107 (Meitner-Programm) and ZK 3 (Zukunftskolleg),
FQXi Large Grant ``The Observer Observed: A Bayesian Route to the 
Reconstruction of Quantum Theory'',
FIS2015-67161-P (MINECO/FEDER),
Basque Government (project IT986-16), and
ERC (Consolidator Grant 683107/TempoQ).




\appendix

\section{Brief introduction to GPTs}\label{app:gpt_def}

In quantum theory the set of effects is represented by Hermitian operators $F$ 
with $0\le F\le \openone$. This convex set has three characteristic properties. 
(i) It is a subset of the real vector space of Hermitian operators.
(ii) There exists the special operator $\openone$ representing the 
all-embracing effect.
(iii) Its shape is given by the partial order $A\le B$ which is defined by the 
condition that $B-A$ is positive semidefinite.

In a GPT, the notion of an effect is generalized by considering a 
straightforward generalization of those properties. We start with an arbitrary 
real vector space $V$ with a partial order $a\le b$. This partial order has to 
be linear in the sense that $a\le b$ implies $\lambda a\le \lambda b$ for any 
$\lambda\in \reals^+$ and $a\le b$ implies $a+c\le b+d$ if also $c\le d$. This 
turns $(V,\le)$ into an ordered vector space.

The all-embracing effect is a distinct element $e\in V$. It is is required to 
dominate all of $V$, i.e., for any $x\in V$ there is a positive number 
$\lambda$ such that $x \le \lambda e$. This property makes $e$ an order unit 
and $(V,\le,e)$ an order unit vector space. In addition, it is convenient to 
assume that the order unit is Archimedean, i.e., if $x\le \lambda e$ holds for 
all $\lambda>0$, then already $x\le 0$. In our paper we implicitly assume that 
any order unit is Archimedean.

It is sometimes convenient to let $V^+= \set{x\in V|0\le x}$. Since $a\le b$ is 
equivalent to $b-a\in V^+$, we then equivalently describe an AOU space by the 
tuple $(V,V^+,e)$. The effects in a GPT are now given by the set $V_e^+= 
V^+\cap (e-V^+)$. A measurement $\M$ in a GPT is represented by a collection of 
elements $\M= (f_k)_k\subset V_e^+$ with $\sum f_k= e$, where $f_k$ represent 
the outcomes of the measurement.

For the set of states, we note that in quantum theory one can represent a state 
$\rho$ equivalently by the linear map $\omega \colon X\mapsto \tr(\rho X)$. 
Then the normalization of $\rho$ becomes $\omega(\openone)=1$ and the condition 
$\rho\ge 0$ reads $\omega(X)\ge 0$ for all $X\ge 0$. By analogy, the set of 
states in a GPT is given by
\begin{equation}
 \ST= \set{ \omega \in V^* |\omega(e)=1 \text{ and } \omega(f)\ge 0 \text{ for 
all }f\ge 0 },
\end{equation}
 where $V^*= \set{\varphi\colon V\rightarrow \reals| \varphi \text{ is 
linear}}$ is the dual space of $V$.
With this definition, the probability for outcome $k$ of a measurement $\M= 
(f_k)_k$ is given by $p_k= \omega(f_k)$.

\section{Bound on $\soc$ for hbits}\label{app:bound}

The proof is by contradiction. Let us assume $\Omega_\mathrm{hbit}=3$, we then 
have $p(01|00)= p(10|10)= p(10|11)=1$, and $p(0|0)= p(1|1)=1$. From $p(0|0)=1$, 
we have $\omega(f_{0|0})= t_{0|0}+\vtr{w}^\dag \vtr{f}_{0|0}=1$, where 
$f_{0|0}= (t_{0|0}, \vtr{f}_{0|0})$. On the other hand, by the definition of 
effects and state, we have $\norm{\vtr{f}_{0|0}} \le \min(t_{0|0}, 1-t_{0|0})$ 
and $\norm{ \vtr{w}}_{*} \le 1 $. We then have
\begin{equation}
t_{0|0} \ge \frac{1}{2} \text{ and } \vtr{w}^\dag \vtr{f}_{0|0}=1-t_{0|0},
\end{equation}
From $f_{0|0} \ne e$ (because $p(00|00)=0$), we have $\vtr{f}_{0|0} \ne 
\vtr{0}$ and  hence $t_{0|0} < 1$.
Then, using again $\norm{ \vtr{w}}_{*} \le 1 $ and $\norm\cdot= \norm\cdot_*$, 
together with Cauchy--Schwarz inequality, we have
\begin{equation}\label{eq:wpar1}
\vtr{w}= \frac{\vtr{f}_{0|0}}{ \norm{\vtr{f}_{0|0}}}, \text{ with } 
\norm{\vtr{f}_{0|0}}=1-t_{0|0}.
\end{equation}
Similarly, we obtain
\begin{equation}\label{eq:wpar2}
\vtr{w}= \frac{\vtr{f}_{1|1}}{\norm{ \vtr{f}_{1|1}}}, \text{ with }
\norm{\vtr{f}_{1|1}}=1 - t_{1|1},
\end{equation}
and, again, $t_{1|1}<1$.

We need now to characterize the terms of the form $\omega(\II_{a|x} f_{b|y})$, 
corresponding to sequences of length two. 
We use the constraints that arise from the condition that the transformation
must map effects to effects. Then, we use that 
$\II_{a|x}$ is a linear transformation that maps the identity element to 
$f_{a|x}$, i.e., $\II_{a|x}e= f_{a|x}$. We, thus, have
\begin{equation}
\II_{a|x}= \left(
 \begin{array}{c|c} t_{a|x} & \vtr{\alpha}_{a|x}^\dag \\ \hline \vtr{f}_{a|x} & 
B_{a|x} \end{array} \right),
\end{equation}
where $\vtr{\alpha}_{a|x}$ is a $n$-dimensional vector and $B_{a|x}$ a $n\times 
n$ matrix. The expectation value 
can then be written as
\begin{eqnarray}
 \omega(\II_{a|x} f_{b|y})= \begin{matrix}\begin{pmatrix}1,& \vtr{w}^\dag 
  \end{pmatrix}\\\mbox{}\end{matrix}
  \begin{pmatrix} t_{a|x} & \vtr{\alpha}_{a|x}^\dag \\ \vtr{f}_{a|x} & B_{a|x} 
  \end{pmatrix} \begin{pmatrix} t_{b|y} \\ \vtr{f}_{b|y} \end{pmatrix}\\
  = \begin{matrix}\begin{pmatrix}t_{a|x} + \vtr{w}\cdot \vtr{f}_{a|x}, & 
  \vtr{w}^\dag B_{a|x}+ \vtr{\alpha}_{a|x}^\dag 
  \end{pmatrix}\\\mbox{}\end{matrix}\begin{pmatrix} t_{b|y} \\ \vtr{f}_{b|y} 
  \end{pmatrix}\\
  = t_{b|y} \left( t_{a|x} + \vtr{w}\cdot \vtr{f}_{a|x} \right) + 
\vtr{f}_{b|y}\cdot (B^\dag_{a|x}\vtr{w} + \vtr{\alpha}_{a|x}).
\end{eqnarray}
We can see the transformation $\II_{a|x}$, applied to the left, as a state 
transformation, i.e., Schr\"odinger picture and with normalization 
corresponding to outcome probability. Then, we have that 
${\norm{B^\dag_{a|x}\vtr{w} + \vtr{\alpha}_{a|x}}_*\le t_{a|x} + \vtr{w}\cdot 
\vtr{f}_{a|x} }$. Notice that such a condition also guarantees that $p(ab|xy) 
\ge 0$ and ${p(ab|xy)\le p(a|x)}$.

This translates to ${\norm{B^\dag_{a|x}\vtr{w} + \vtr{\alpha}_{a|x}}_* \le 1}$ 
for the case $(a,x)= (0,0)$ or $(1,1)$. In fact, in those cases we have 
$t_{a|x} + \vtr{w}\cdot \vtr{f}_{a|x}=1$, so the dual norm condition guarantee 
that $\omega(\II_{a|x} f_{b|y})\le 1$ for all $f_{b|y}$.

From the conditions $p(11|11)=0$, we obtain
\begin{eqnarray}
t_{1|1} \left( t_{1|1} + \vtr{w}\cdot \vtr{f}_{1|1} \right) + 
\vtr{f}_{1|1}\cdot (B_{1|1}^\dag\vtr{w} + \vtr{\alpha}_{1|1}) \\
 = t_{1|1} + \vtr{f}_{1|1}\cdot (B^\dag_{1|1}\vtr{w} + \vtr{\alpha}_{1|1})=0,
\end{eqnarray}
which implies together with Eq.~\eqref{eq:wpar2}, ${\norm{B^\dag_{1|1}\vtr{w} + 
\vtr{\alpha}_{1|1}}_*\le 1}$, and Cauchy--Schwarz inequality, that 
\begin{equation}\label{eq:bw}
{(B^\dag_{1|1}\vtr{w} + \vtr{\alpha}_{1|1})= -\vtr{w}} \text{ and }t_{1|1}= 
\frac{1}{2}.
\end{equation}

On the other hand, we have $p(10|10)=1$ that, by Eq.~\eqref{eq:wpar1} and 
Eq.~\eqref{eq:bw}, implies
\begin{eqnarray}
t_{0|0} \left( t_{1|1} + \vtr{w}\cdot \vtr{f}_{1|1} \right) + 
\vtr{f}_{0|0}\cdot (B_{1|1}^\dag\vtr{w} + \vtr{\alpha}_{1|1}) \\ \nonumber
 = t_{0|0} - \vtr{f}_{0|0}\cdot \vtr{w}=2 t_{0|0} - 1=1,
\end{eqnarray}
which implies $t_{0|0}=1$, i.e., a contradiction with $t_{0|0}<1$, which 
concludes the proof.

\section{Capacity of dichotomic norm cones}\label{app:dim}

In the following we prove that dichotomic norm cones describe systems of capacity two. For convenience, we repeat Eq.~\eqref{eq:def_dim} from the main 
text.
\begin{equation*}
\sum_k f_k \le e\text{ and } \omega_i f_j= \delta_{ij} \text{ for all }
 i,j.
\end{equation*}

We first show that a capacity of two is an upper bound.
\begin{lemma}
In a dichotomic norm cone, let $(\omega_k)_k$ be a collection of $d$ states and 
$(f_k)_k$ a collection of $d$ effects, such that Eq.~\eqref{eq:def_dim} is 
satisfied.
Then $d\le 2$.
\end{lemma}

\begin{proof}
Any effect $f= (t,\vtr x)$ must satisfy $0\le f\le e$, i.e., $t\ge \norm{\vtr 
x}$ and $1-t\ge \norm{\vtr x}$. Furthermore, a state $\omega\colon (s,\vtr 
y)\mapsto s+\vtr w^\dag \vtr y$ must obey $\norm{\vtr w}_*\le 1$. It follows 
that $\vtr w^\dag \vtr x\le \norm{\vtr x}$ and hence $\omega f=1$ requires 
$t\ge \frac12$. Thus $\sum_k f_k\le e$ implies for $f_k= (t_k,\vtr x_k)$ the 
inequalities
\begin{equation}
 0\le \norm{\sum_k \vtr x_k}\le 1-\sum_k t_k \le 1-\frac d2.
\end{equation}
Which yields at once the assertion.
\end{proof}

In addition, if the dimension of the underlying vector space is finite, we can 
always find vectors $\vtr x$ and $\vtr w$, such that $\norm{\vtr x}=1$, 
$\norm{\vtr w}_*=1$, and $\vtr w^\dag \vtr x=1$. Hence, the states 
$\omega_{1,2}= (1,\pm \vtr w)$ and effects $f_{1,2}= (1,\pm \vtr x)/2$ obey 
Eq.~\eqref{eq:def_dim}. It follows that the capacity of a dichotomic norm cone 
is always exactly two.

\bibliography{mem_cost}

\end{document}